\documentclass[english]{cccconf}
\usepackage[comma,numbers,square,sort&compress]{natbib}
\usepackage{epstopdf, xcolor, xurl}
\renewcommand{\emph}[1]{\textit{#1}}
\definecolor{NCSUred}{RGB}{153, 0, 0}
\definecolor{NCSUgreen}{RGB}{0, 132, 115}
\definecolor{NCSUblue}{RGB}{65, 86, 161}
\definecolor{NCSUorange}{RGB}{209, 73, 5}

\usepackage{hyperref}
\hypersetup{colorlinks, breaklinks=true, plainpages=false, citecolor=NCSUblue, linkcolor=NCSUblue, urlcolor=NCSUgreen, bookmarksopen=true, bookmarksnumbered=false, bookmarksdepth=5}
\usepackage{amsmath, amsthm, amsfonts, amssymb}
\usepackage{dsfont}
\newtheorem{theorem}{Theorem}
\newtheorem{lemma}{Lemma}
\newtheorem{corollary}{Corollary}
\newtheorem{proposition}{Proposition}
\newtheorem{definition}{Definition}

\newtheorem{example}{Example}

\newcommand{\mc}[1]{\mathcal{#1}}

\newcommand{\mr}[1]{\mathrm{#1}}
\newcommand{\mbb}[1]{\mathds{#1}} % for the consistency of fonts

\newcommand{\mR}{\mbb{R}}
\newcommand{\mC}{\mbb{C}}
\newcommand{\mN}{\mbb{N}}
\newcommand{\mD}{\mbb{D}}
\newcommand{\mX}{\mbb{X}}

\newcommand{\bra}[1]{\left( #1 \right)}
\newcommand{\Bra}[1]{\left[ #1 \right]}
\newcommand{\BRA}[1]{\left\{ #1 \right\}}
\newcommand{\norm}[1]{\left\| #1 \right\|}
\newcommand{\ip}[2]{\langle #1, \, #2 \rangle}

\usepackage{subcaption}

\begin{document}

\title{Koopman-based Estimation of Lyapunov Functions: \\ Theory on a Reproducing Kernel Hilbert Space}
\author{Wentao Tang\aref{ncsu}, Xiuzhen Ye\aref{ncsu}}
\affiliation[ncsu]{Department of Chemical and Biomolecular Engineering, North Carolina State University, Raleigh 27695, North Carolina, U.S.A. 
\email{wtang23@ncsu.edu; xye7@ncsu.edu}}
\maketitle

\begin{abstract}
    Koopman operator provides a general linear description of nonlinear systems, whose estimation from data (via extended dynamic mode decomposition) has been extensively studied. 
    However, the elusiveness between the Koopman spectrum and the stability of equilibrium point poses a challenge to utilizing the Koopman operator for \emph{stability analysis}, which further hinders the construction of a universal theory of \emph{Koopman-based control}. 
    In our prior work, we defined the Koopman operator on a reproducing kernel Hilbert space (RKHS) using a linear--radial product kernel, and proved that the Koopman spectrum is confined in the unit disk of the complex plane when the origin is an asymptotically stable equilibrium point. 
    Building on this fundamental spectrum--stability relation, here we consider the problem of Koopman operator-based Lyapunov function estimation with a given decay rate function. The decay rate function and the Lyapunov function are both specified by positive operators on the RKHS and are related by an operator algebraic Lyapunov equation (ALE), whose solution exists uniquely. 
    The error bound of such a Lyapunov function estimate, obtained via kernel extended dynamic mode decomposition (kEDMD), are established based on statistical learning theory and verified by a numerical study. 
\end{abstract}
\keywords{Nonlinear systems, Koopman operator, Lyapunov function, Reproducing kernel Hilbert space}
\footnotemark\footnotetext{This work is supported by Prof. Wentao Tang's faculty startup funds from North Carolina State University. }

\section{Introduction}
\par Nonlinear systems are ubiquitous in engineering applications, which, however, often lack accurate models from first principles. The concept of \emph{Koopman operator}, originated in the studies of statistical physics, has become a standard tool of nonlinear systems theory \cite{mauroy2020koopman}.  
Considering the evolution of state-dependent functions, the Koopman operator maps any such function to its composition with the flow of the dynamics. Whenever the function space in consideration remains invariant upon this composition, the Koopman operator is linear, and thus conceptually provides a ``global linearization'' of the nonlinear system. 
As a linear operator, the Koopman operator can be \emph{learned from data} in a computationally efficient convex optimization formulation \cite{william2015data, klus2020eigendecompositions, kostic2022learning}. This brings hope for a generic ``data-driven'' control strategy for nonlinear systems \cite{tang2022data}. 

\par To the end of data-driven Koopman modeling, the \emph{proper choice of a function space}, which must be \emph{invariant} under the action of Koopman operator, becomes an issue not only for theoretical rigor but also for its computation and application in control. 
For example, in \cite{strasser2025koopman}, based on the idea that Koopman operator converts a controlled nonlinear system into a high-dimensional bilinear system, a linear robust controller synthesis was proposed. To guarantee control performance, it is assumed that the Koopman operator is invariant on a finite-dimensional space spanned by some basis functions. Similar thoughts are followed in our recent works on Koopman-based state observers \cite{ye2025edmd, ni2025data}. 
In fact, if the Koopman operator maps outside of its domain, the resulting error analysis must be restricted to instantaneous (one-sampling-time) prediction, thus invalidating its prediction over long horizons \cite{philipp2025error} and hence any stability guarantee. 

\par Unfortunately but expectedly, no function space is universally suitable for the definition of Koopman operators; instead, the choice depends on the qualitative behaviors of nonlinear systems. 
For a measure-preserving dynamics, namely the ergodic motion on an attractor set, one defines the Koopman operator as an isometry on $L^2(\mu)$ (where $\mu$ denotes the invariant measure) \cite{korda2018convergence, colbrook2024rigorous}. For an asymptotically stable dynamics towards an equilibrium point, the Koopman operator can be defined on a Segal--Bargmann space (which is a subspace of the collection of analytical functions) \cite{mezic2020spectrum}. 
For the convenience of data-based estimation, however, one should rather define the Koopman operator on a \emph{reproducing kernel Hilbert space (RKHS)} \cite{paulsen2016introduction} and thus use kernel-based learning methods without choosing basis functions \cite{william2015kernel, kawahara2016dynamic}. 
The RKHS theory of Koopman operators is only recently noted in the literature --- as long as the dynamics has sufficient smoothness, the Koopman operator can be accommodated in an RKHS, if choosing a kernel that makes the RKHS equivalent to a Sobolev-type Hilbert space \cite{kohne2025Linf, tang2025koopman, tang2025data} --- see Section \ref{sec:preliminaries} for an elementary introduction to this theory. 

\par While Koopman operator gives a linear representation of nonlinear systems, the use of Koopman operator for dynamical analysis and control largely remains an art of practicality rather than a theoretically rigorous framework. Intrinsically, the property of an operator on an (infinite-dimensional) function space differs from that of a matrix on a (finite-dimensional) Euclidean space. 
As a symptomatic fact, if the spectral radius of $A$ is below $1$, then the origin is asymptotically stable \cite{antsaklis2006linear}; but when $A$ stands for the Koopman operator $A$, the spectrum of $A$ depends on the choice of the kernel function in RKHS, and the \emph{fundamental relation between the spectrum and stability can be elusive}. 
As commonly pointed out, despite many exploratory efforts, a general theory for Koopman-based control does not yet exist \cite{strasser2025overview, haseli2025roads}. Clearly, to this end, the spectrum--stability relation on RKHS is of pivotal importance. 
This issue was addressed in the authors' very recent work (Tang and Ye \cite{tang-ye-2025-l4dc}). By defining a linear--radial product kernel, we proved that the spectral radius of $A$ is below $1$ when the origin is asymptotically stable, while upon local bifurcation that alters the stability, the spectrum of $A$ escapes from the unit disk --- see Section \ref{sec:kernel} for more detailed theoretical discussions. 

\par The aim of this paper is to utilize the fundamental relation between the Koopman operator's spectrum, previously established in \cite{tang-ye-2025-l4dc}, for Lyapunov stability analysis. 
\begin{itemize}
    \item Considering the problem of finding a Lyapunov function with a fixed decay rate, we represent the decay rate function and the Lyapunov function as ``\emph{kernel quadratic forms}'' represented by positive operators on the RKHS, thus formulating a (discrete-time) \emph{operator algebraic Lyapunov equation} (operator ALE). 
    \item We show that when the spectral radius of the Koopman operator is below $1$, the operator ALE has a unique solution expressed by a convergent infinite series of Hilbert--Schmidt operators. (See Section \ref{sec:Lyapunov}.)
    \item When the Koopman operator must be estimated from data by a kernel extended dynamic mode decomposition (kEDMD) routine, we prove that the resulting Lyapunov function estimate can have an error at most in proportion to the learning error of the Koopman operator, thus justifying a data-based routine to calculate the Lyapunov function. (See Section \ref{sec:estimation}.)
\end{itemize}
The proposed approach is demonstrated with a numerical study on a non-polynomial nonlinear system (Section \ref{sec:numerical}), where the Lyapunov function is shown to be approximately correctly found. % The shape of the Lyapunov function not only embodies a quadratic form due to the locally linearized dynamics near the origin, but also exhibits significant deviation from such a linearized dynamics faraway from the origin, reflecting the underlying severe nonlinearity. 

\section{Preliminaries}\label{sec:preliminaries}
For a discrete-time dynamical system 
\begin{equation}\label{eq:system}
    x_{t+1} = f(x_t), \enspace f: \mX\rightarrow \mX \subset \mR^{d}, 
\end{equation}
its \emph{Koopman operator} is defined as the following linear mapping on function space $\mc{G}$: 
\begin{equation}\label{eq:Koopman}
    A: \mc{G}\rightarrow \mc{G}, g \mapsto g \circ f. 
\end{equation}
We say that the Koopman operator is \emph{well-defined} if for any state-dependent function $g\in \mc{G}$, we have $Ag\in \mc{G}$. As such, Koopman operator captures the nonlinear system \eqref{eq:system} in a function space, which typically is not finite-dimensional. 
\par An intuitive approximation scheme for the Koopman operator $A$, initially by \cite{william2015data}, is extended dynamic mode decomposition (EDMD), where the $d$-dimensional state $x$ is lifted by a large number of ($N$) dictionary functions $z=(\psi_1(x), \dots, \psi_N(x))$ and a $N$-dimensional linear system $z_{t+1} \approx \hat{A}_Nz_t$ is estimated through regression. 
However, the choice of dictionary $\psi$ and in what sense it forms an approximated basis of $\mc{G}$ tend to be difficult to capture rigorously.

\subsection{Hilbert Space and Sobolev--Hilbert Space}
\par Suppose that $\mc{G}$ can be chosen to be a Hilbert space, i.e., on which an inner product $\ip{g_1}{g_2}$ can be defined and the induced norm $\|g\|=\sqrt{\ip{g}{g}}$ makes $\mc{G}$ complete. 
If $\mc{G}$ is further \emph{separable}\footnote{That is, for any $\epsilon>0$ and any $g\in \mc{G}$, there exists a finite number of $h_1,\dots,h_N\in \mc{G}$, $N$ dependent on $\epsilon$ and $g_\epsilon \in \mr{span}\{h_i\}_{i=1}^N$, such that $\|g-g_\epsilon\|<\epsilon$. Such a set of $\{h_i\}_{i=1}^N$ is said to be an $\epsilon$-net.} Then it is known from functional analysis (e.g., \cite{zhang2021lectures}) that there exist a \emph{countable} basis $\{u_i\}_{i\in \mN}$ such that any $g\in \mc{G}$ can be expanded as a Fourier series $g = \sum_{i\in \mN} c_iu_i$ and $\|g\|^2 = \sum_{i\in \mN} |c_i|^2$. 
In this sense, $\mc{G}$ is isometrically isomorphic to $\ell^2(\mN)$. Naturally the Koopman operator $A$, if well-defined on such a space $\mc{G}$, would be convenient to learn on a basis $\{e_i\}_{i\in \mN}$ or an $\epsilon$-net.

\par A well-known type of Hilbert space is the Sobolev--Hilbert space $W^{s, 2}(\mX)$, which is the space of functions $g$ such that $g$ has generalized derivatives up to the $s$-th order that are all square-integrable on $\mX$. 
\footnote{When $s$ is not integer, the extended definition of $W^{s,2}(\mX)$, also known as the Sobolev--Slobodeckij spaces, can be found in  \cite{zhang2025introduction}. Correspondingly, $C^s(\mX)$, as the class of functions that are continuously differentiable up to the $s$-th order, can be defined for non-integer $s$ also in a generalized manner.} 
As proven in K{\"{o}}hne et al. \cite{kohne2025Linf}, the Koopman operator is \emph{well-defined as a bounded linear operator on a Sobolev--Hilbert space $W^{s,2}(\mX)$}, under the following mild conditions: 
(i) $\mX\subset \mR^d$ is compact, (ii) the dynamics $f$ is $C^s$ on $\mX$, and (iii) $f$ is non-degenerate, in the sense of $\inf_{x\in \mX} \lvert \mr{det}~\mr{D}f(x) \rvert > 0$. 

\par However, in general, it is still difficult to find a basis for $W^{s,2}(\mX)$ or an $\epsilon$-net. 
In machine learning, a standard technique to formulate the nonlinear function learning as (computationally tractable) convex optimization problems (as if the function is ``linear'') is the \emph{kernel method} \cite{steinwart2008support}. 
Its formally rigorous description lies in the theory of RKHS.

\subsection{Reproducing Kernel Hilbert Space (RKHS)} 
We say that function $\kappa \in C(\mX\times \mX, \mR)$ is a \emph{kernel} if for any $n\in \mN$ and $\{x^{(i)}\}_{i=1}^n \subset \mbb{X}$, the matrix $G_\kappa = \Bra{ \kappa\bra{x^{(i)}, x^{(j)}} }_{i,j=1}^n$ is symmetric and positive semidefinite. 
The RKHS specified by kernel $\kappa$ on $\mX$ is defined as
$$H_\kappa(\mX) = H_\kappa = \overline{\mr{span}}\BRA{\kappa(x, \cdot): x\in \mX}$$
on which the inner product is defined by $\ip{\kappa(x,\cdot)}{\kappa(x',\cdot)} = \kappa(x,x')$ and hence the norm: $\|\kappa(x, \cdot)\| = \kappa(x,x)^{1/2}$. 
On the RKHS, the \emph{reproducing property} refers to the fact that $\forall g\in H_\kappa$, $\ip{g}{\kappa(x, \cdot)} = g(x)$. The RKHS is separable when $\mX \subset \mR^d$ \cite{steinwart2008support}. For any $x\in \mX$, the kernel function at $x$, i.e., $\kappa(x,\cdot) \in H_\kappa$, is viewed as an ``infinite-dimensional'' lifting of $x$ into $H_\kappa$ and called the \emph{canonical feature} of $x$. 

\par However, the postulation that Koopman operator is well-defined on any RKHS is not always guaranteed. Nonetheless, the following lemma from Wendland \cite{wendland2004scattered} establishes that a well-chosen kernel can make the resulting RKHS coincide with a Sobolev--Hilbert spaces. 
\begin{lemma}\label{lem:Sobolev-RKHS}
    Suppose that $\kappa$ is a \emph{radial kernel}, namely one such that $\kappa(x, x') = \rho(|x-x'|)$ for some $\rho: \mR_+\rightarrow \mR_+$, and that the Fourier transform of $\rho$, $\hat{\rho}(\xi)$, satisfies $c_1(1+|\xi|^2)^{-s/2}\leq |\hat{\rho}(\xi)| \leq c_2(1+|\xi|^2)^{-s/2}$ for some constants $c_2\geq c_1>0$.
    \footnote{According to \cite{wendland2004scattered}, a kernel that satisfies these conditions can be found by using the following radial function $\rho$. 
    For any $k\in \mbb{N}\cup\{0\}$ and $s=\frac{d}{2}+k$, we can let $\rho = I^k \varrho_{\lfloor d/2+k+1 \rfloor}$, 
    where $\varrho_\ell(r) = \max\{1-r, 0\}^\ell$ and the operator $I$ is defined by $I\phi(r) = \int_r^\infty r'\phi(r')dr'$.}
    In addition, assume that $\mbb{X}$ is a region with Lipschitz boundary. Then $W^{s,2}(\mbb{X})=H_\kappa(\mbb{X})$, with equivalent norms. 
\end{lemma}

From now on, we denote such a kernel $\kappa$ by ``$\kappa^=$'' (the rationale of superscript ``$=$'' is explained in the next footnote), without explicitly stating the index $k$ of the kernel. 
\begin{corollary}\label{cor:Wendland}
    Suppose that (i) $\mbb{X}$ is a bounded region with Lipschitz boundary, (ii) $f\in C^s(\mbb{X})$ with $s=\frac{d}{2}+k$, $k\in \mbb{N}$, and (iii) $\inf_{x\in \mbb{X}} \lvert \mr{D}f(x) \rvert > 0$. Then $A$ is a bounded linear operator on $H_{\kappa^=}(\mbb{X})$, where $\kappa^=$ is specified by the conditions in Lemma \ref{lem:Sobolev-RKHS}. 
\end{corollary}

\par When $A: H_\kappa\rightarrow H_\kappa$ is well-defined and bounded, its adjoint operator $A^\ast: H_\kappa\rightarrow H_\kappa$, specified by the defining property, $\ip{g}{Ah} = \ip{A^\ast g}{h}$ ($\forall g, h\in H_\kappa$), is also bounded. It then turns out that the $A^\ast$ is the operator that ``pushes the canonical feature forward in time''. That is, 
\begin{equation}\label{eq:push-forward}
  A^\ast \kappa(x, \cdot) = \kappa(f(x), \cdot), \, \forall x\in \mbb{X}.  
\end{equation}
This property will be useful for the Koopman operator learning, which is to be discussed in Section \ref{sec:estimation}.

\section{Linear--Radial Product Kernel}\label{sec:kernel}
While a radial kernel suffices to guarantee the well-definedness of Koopman operator under regularity conditions, it is ``indifferent to'' or ``unaware of'' the existence of an equilibrium point, and hence the Koopman operator, in particular its spectrum, will be noninformative of stability. 
To see this point, in view of the relation \eqref{eq:push-forward}, we have $(A^*)^t\kappa^=(x, \cdot) = \kappa^=(f^t(x), \cdot)$ in which $\|\kappa(x, \cdot)\| = \rho(0)^{1/2} = \|\kappa(f^t(x), \cdot)\|$, and hence $\|A^t\|\geq 1$ for all $t$. This implies that the spectral radius $r(A) = \limsup_{t\rightarrow \infty}\|A^t\|^{1/t}\geq 1$, regardless of whether the origin is a stable or unstable equilibrium point. 

\par To address this theoretical gap, our prior work \cite{tang-ye-2025-l4dc} proposed a product kernel $\kappa = \kappa^-\kappa^=$ comprising of a linear kernel $\kappa^-(x,x') = x\cdot x'$ and a Wendland radial kernel $\kappa^=(x,x') = \rho(|x-x'|)$ as specified by Lemma \ref{lem:Sobolev-RKHS}.\footnote{To explain the rationale of the symbols, the superscript ``$-$'' indicates ``linearity'' of a linear kernel, and ``$=$'' stands for ``translation invariance'', namely the radial property of the Wendland kernel. Also, ``$-$'' and ``$=$'' mean ``the first'' and ``the second'' kernel, respectively.} 
The two RKHSs are, respectively, $H_{\kappa^-}(\mX) = \{c^\top x: c\in \mX\}$, which is further equal to $\{c^\top x: c\in \mR^d\}$ if $0\in \mr{int}(\mX)$, and $H_{\kappa^=} = W^{s,2}(\mX)$, which captures the regularity of the dynamics. 
As such, the following lemma gives a general depiction of the members of such a novel RKHS. 
Simply speaking, the RKHS specified by a \emph{linear--radial product kernel} comprises of functions that appear to be \emph{locally linear or higher-order} near the origin (equilibrium point), and \emph{globally regular} to the $s$-th-order derivatives. 
\begin{lemma}[Tang \& Ye \cite{tang-ye-2025-l4dc}, Lemma 5]
\label{lem:member}
    If $f(0)=0$, $0\in \mr{int}(\mbb{X})$, and $f\in C^s(\mbb{X})$ with $s=\frac{d}{2}+k$ and $k\in \mbb{N}$, then
    \begin{equation}\label{eq:member}
    \textstyle
        H_\kappa = \BRA{\sum_{i=1}^d e_ih_i: h_1,\dots,h_d\in W^{s,2}(\mbb{X})}.
    \end{equation}
    in which $e_i(x) = x_i$, $i=1,\dots,d$. Moreover, for any $g = \sum_{i=1}^d e_ih_i \in H_\kappa$, it holds that
    $ \norm{g}_{H_\kappa}^2 = \sum_{i=1}^d \norm{h_i}_{H_{\kappa^=}}^2 $.
\end{lemma}

\subsection{Koopman Operator on the RKHS}
Next, we provide the conditions for the Koopman operator $A$ to be well-defined and bounded on $H_\kappa$. Recalling that $f\in C^s(\mX, \mR^d)$ is desired for well-definedness and boundedness on $H_{\kappa^=}= W^{s,2}(\mX)$, we replace $C^s$ by a new class:
\begin{equation}\label{eq:F.class}
\textstyle
    \Phi^s(\mX) := \BRA{\sum_{i=1}^d e_i\phi_i: \phi_1, \dots,\phi_d\in C^s(\mX, \mR)}. 
\end{equation}
Informally, we say that this class of functions is the space of ``$s$-smooth functions that have zero values at the origin''. 
\begin{lemma}[\cite{tang-ye-2025-l4dc}, Theorem 6]
\label{lem:definedness}
    If (i) $\mbb{X}\subset\mR^d$ is a compact Lipschitz region, (ii) $f_1, \dots, f_d \in \Phi^s(\mbb{X})$, and (iii) $f$ is non-degenerate: $\inf_{x\in \mbb{X}} \lvert \mr{det}~\mr{D}f(x) \rvert > 0$. Then the Koopman operator $A$ is a bounded linear operator on $H_\kappa$.
\end{lemma}

\subsection{Spectrum--Stability Relation}
For a finite-dimensional linear systems represented by a matrix $A\in \mR^{d\times d}$, its spectrum (i.e., set of eigenvalues) is fundamentally related to the stability. Now for the Koopman operator $A$ on RKHS, such a spectrum--stability relation is still desired. 
In particular, we hope that the Koopman spectrum $\sigma(A)\subset \mD = \{\lambda \in \mC: |\lambda|<1\}$ when the origin the asymptotically stable.\footnote{By the spectrum of an operator, $\sigma(A)$, we refer to the set of complex numbers $\lambda\in \mC$ such that $\lambda I - A$ does not have a bounded inverse ($I$: the identity operator on $H_\kappa$). In general, the spectrum may not even may be a discrete set, since an operator can have so-called continuous spectrum and residual spectrum.} 
The characterization of $\sigma(A)$ typically requires certain \emph{homeomorphism} conditions \cite{mezic2020spectrum}. By a homeomorphism, we refer to an continuously invertible map $\psi: \mbb{X}\rightarrow \mbb{Z} = \psi(\mbb{X})$. Thus, any property of $f$ is embodied in that of the conjugated system $\tilde{f} = \psi\circ f \circ \psi^{-1}$. 
We are particularly interested in the homeomorphism that carries $f$ to $\tilde{f}: z\mapsto Fz$, where $F = \mr{D}f(0)$ is the \emph{Jacobian}. 

\par The key condition that enables the analysis is the existence of a \emph{$\Phi^s$-homeomorphism} $\psi$, namely a homeomorphism $\psi\in (\Phi^s(\mbb{X}))^d$ such that $\psi^{-1}\in (\Phi^s(\mbb{Z}))^d$. %This requires that in $\mX$, there cannot exist any other invariant structure than the equilibrium point at $0$. 
Under this condition, it turns out that the spectrum of $A$ is a semigroup generated by the eigenvalues of the Jacobian $F = \mr{D}f(0)$. 
\begin{lemma}[\cite{tang-ye-2025-l4dc}, Theorem 8]
\label{lem:homeomorphism}
    Let the conditions in Lemma \ref{lem:definedness} hold and there be a $\Phi^s$-homeomorphism $\psi$ such that $\tilde{f} = \psi\circ f\circ \psi^{-1}: z\mapsto Fz$. 
    Then $\textstyle r(A) \leq r(F)$, i.e., the Koopman spectral radius is bounded by the maximum modulus among the Jacobian eigenvalues. Specifically, any element of $\sigma(A)$ is a product of Jacobian eigenvalues or a limit of such products. 
    In particular, if $F$ is stable, then $r(A)<1$. 
\end{lemma}
\begin{example}\label{example1}
    For the Brusselator model in continuous time: $$\dot{x}_1 = a+x_1^2x_2 - (b+1)x_1, \enspace \dot{x}_2 = bx_1 - x_1^2x_2,$$
    when $a>0$ and $0<b<1+a^2$, there exists a unique asymptotically stable equilibrium point at $(a, b/a)$. For $a=1$ and $b=1$, the eigenvalues of Jacobian are $-\frac{1}{2}\pm \frac{\sqrt{3}}{2}\mr{i}$. 
    For the discrete-time dynamics with a sampling interval $\Delta$, the eigenvalues of Jacobian are $\exp\bra{-\frac{1}{2}\Delta \pm \frac{\sqrt{3}}{2}\Delta\mr{i}}$. This dynamics satisfy all the conditions in Lemma \ref{lem:homeomorphism} when $\mX\subset \mR^2$ is any compact invariant set. Hence we have 
    \begin{equation}\label{eq:exact.spectrum}
        \textstyle 
        \sigma(A) = \{0\}\cup \BRA{ \mr{e}^{-(p+q)\frac{\Delta}{2} + \mr{i}(p-q)\frac{\sqrt{3}\Delta}{2}} : p,q\in \mN } \backslash \{1\}
    \end{equation}
    and in particular $r(A) = \mr{e}^{-\frac{1}{2}\Delta}<1$. 
\end{example}

\section{Lyapunov Function as a Kernel Quadratic Form}\label{sec:Lyapunov}
As the main objective of this paper, we discuss in this section how a Lyapunov function can be obtained from the Koopman operator, when the latter has a spectrum contained in $\mD$ provided that the origin is an asymptotically stable equilibrium point on $\mX$. 
It is well known that for a finite-dimensional linear system, the Lyapunov function can be written as a quadratic form $v(x) = x^\top Px$ where $P$ is a positive semidefinite matrix. 
As an extension to the RKHS $H_\kappa$, where the Lyapunov function can be more general, we postulate a ``\emph{kernel quadratic form}'' $v(x) = \ip{\kappa(x,\cdot)}{P\kappa(x, \cdot)}$ specified by some \emph{positive operator} $P: H_\kappa \rightarrow H_\kappa$. 
\footnote{By a positive operator $P$ we refer to the property that $\ip{g}{Pg}\geq 0$ ($\forall g\in H_\kappa$). 
We write $Q\geq 0$ to indicate that $Q$ is a positive operator and $Q\geq R$ to indicate that $Q-R\geq 0$.}
For subsequent discussions, we review the concepts of \emph{Hilbert--Schmidt (HS) operators} \cite{zhang2021lectures}. 
\begin{definition}
    A symmetric operator $Q=Q^\ast$ on a Hilbert space $H$ is said to be \emph{Hilbert--Schmidt}, denoted by $Q\in \mr{HS}(H)$, if under any orthonormal basis $\{u_i\}_{i\in \mN}$ of $H$, 
    \begin{equation}\label{eq:HS}
        \textstyle
        Q = \sum_{i\in\mN} q_iu_i\times u_i
    \end{equation}
    for some $\{q_i\}_{i\in\mN} \subset\mR$ such that $\sum_{i\in \mN} q_i^2<\infty$. \footnote{Here $g\times h$ refers to a rank-$1$ operator defined by $(g\times h)k = \ip{h}{k} g$ for any $g, h, k\in H$. Clearly, $u_i\times u_i \geq 0$ since for any $g\in H$, $\ip{g}{(u_i\times u_i)g} = |\ip{u_i}{g}|^2 \geq 0$. The choice of orthonormal basis is non-influential.}
\end{definition}
\begin{proposition}\label{prop:HS}
    Let $H$ be a Hilbert space. Then
    \begin{enumerate}
        \item $\mr{HS}(H)$ is a Hilbert space of operators. Specifically, for any $Q, R\in \mr{HS}(H)$, $\ip{Q}{R}_{\mr{HS}} = \sum_{i\in\mN} q_ir_i$ and $\|Q\|_{\mr{HS}}^2 = \sum_{i\in\mN} q_i^2$. 
        \item $Q\in \mr{HS}(H)$ is a positive operator if and only if $q_i \geq 0$ for all $i\in \mN$ in its above representation. 
        \item Finite-rank operators are dense in the set of HS operators. That is, for any $\epsilon>0$ and any $Q\in \mr{HS}(H)$, there exists a $N\in \mN$ and $Q_N$ of rank not exceeding $N$, such that $\|Q-Q_N\|_{\mr{HS}}<\epsilon$. 
    \end{enumerate}
\end{proposition}

\subsection{Operator Algebraic Lyapunov Equation} 
For finite-dimensional linear systems, we know that the $P$ matrix specifying the Lyapunov function can be obtained by solving an \emph{algebraic Lyapunov equation} (ALE), which has an unique positive definite solution if the origin is asymptotically stable. 
For nonlinear system \eqref{eq:system}, we consider the problem of finding a Lyapunov function $v(\cdot)$ that satisfies 
\begin{equation}\label{eq:decay}
    v(f(x)) - v(x) \leq -w(x)
\end{equation}
for some decay rate function $w\geq 0$. 

\par We focus on the case where $$w(x) = \ip{\kappa(x, \cdot)}{Q\kappa(x, \cdot)}$$ for some positive $Q\in \mr{HS}(H_\kappa)$. Given the representation \eqref{eq:HS}, where all $q_i\geq0$ and $q_i\downarrow 0$ (without loss of generality), we have $w(x) = \sum_{i\in \mN}q_iu_i(x)^2$. 
Here $\{u_i\}$ forms a basis of $H_\kappa$; hence, according to Lemma \ref{lem:member}, $u_i = \sum_{j=1}^d e_jh_{ij}$ for some $h_{ij}\in W^{s,2}(\mX)$. In this sense, we say that $w$ is an ``approximate sum of squares'' of ``locally linear'' and ``globally regular'' functions. 
For example, if $w(\cdot) = |\cdot|^2 = \sum_{i=1}^d e_i^2$, then $Q = \sum_{i=1}^d e_i\times e_i$ is simply a rank-$d$ operator. 

\par Clearly, $\kappa(0, \cdot) = 0$, $w(0)=0$, and thus $v(0)=0$ for any solution of \eqref{eq:decay}. Then, by doing a telescopic sum of \eqref{eq:decay} and bringing to the infinite-time limit, as $\lim_{t\rightarrow\infty} f^t(x)\rightarrow 0$ for all $x\in X$, we obtain
$v(x) = \sum_{t=0}^\infty w\bra{f^t(x)} = \sum_{t=0}^\infty \ip{A^{*t}\kappa(x, \cdot)}{QA^{\ast t}\kappa(x, \cdot)}. $
That is, 
\begin{equation}\label{eq:P.series}
    \textstyle 
    v(x) = \ip{\kappa(x, \cdot)}{P\kappa(x, \cdot)}, \enspace 
    P = \sum_{t=0}^\infty A^t Q A^{*t}.
\end{equation}
The operator $P$ here satisfies the following \emph{operator algebraic Lyapunov equation (operator ALE)}:
\begin{equation}\label{eq:P.ale}
    APA^\ast - P = -Q. 
\end{equation}

\subsection{Existence and Uniqueness of the Solution}
Now we present the main theoretical contribution of this paper, which serves as a \emph{converse Lyapunov theorem} in the language of a Koopman operator on a RKHS. 
\begin{theorem}
    Let $H$ be a Hilbert space, $Q\in \mr{HS}(H)$ a positive operator, and $A: H\rightarrow H$ a bounded linear operator with $r(A)<1$, i.e., $\sigma(A)\subset \mD$. Then the the series given in \eqref{eq:P.series} determines a $P\in \mr{HS}(H)$ as a solution to the operator ALE \eqref{eq:P.ale} and in fact the unique solution.  
\end{theorem}
\begin{proof}
    Let $\{u_i\}_{i\in \mN}$ be an orthonormal basis of $H$ and $Q$ be represented as \eqref{eq:HS}. Then for $t\in \mN$, we have
    \begin{align*} 
    \textstyle 
        &\|A^tQA^{\ast t}\|_{\mr{HS}}^2 = \|\sum_{i\in \mN} q_i A^tu_i\times A^tu_i \|_{\mr{HS}}^2 \\
        &= \sum_{i\in \mN}\sum_{j\in \mN} q_iq_j|\ip{A^tu_i}{A^tu_j}|^2 \leq \|Q\|_{\mr{HS}}^2\|A^t\|^4. 
    \end{align*}
    and hence $\|A^tQA^{\ast t}\|_{\mr{HS}} \leq \|Q\|_{\mr{HS}}\|A^t\|^2$. Given any $\epsilon \in (0, (1-r(A))/2)$, there exists a $\overline{t}\in \mN$ such that $\|A^t\|\leq (r(A)+\epsilon)^t$ for all $t\geq \overline{t}$, and hence $\sum_{t=\overline{t}}^\infty A^tQA^{\ast t}$ is an HS operator with HS norm bounded by $\|Q\|_{\mr{HS}} \sum_{t=\overline{t}}^\infty (r(A)+\epsilon)^{2t} \leq (r(A)+\epsilon)^{-2\overline{t}}$. 
    Clearly, $\sum_{t=0}^{\overline{t}-1} A^tQA^{\ast t}$ is HS. Thus the series $P$ is HS. To prove uniqueness, suppose that $P$ and $P'$ are two solutions to \eqref{eq:P.ale}, then $\tilde{P} = P-P'$ satisfies $A\tilde{P}A^\ast = \tilde{P}$. It follows that $\tilde{P} = A^t\tilde{P}A^{\ast t}$ for all $t\in \mN$. Hence we have, 
    $\|\tilde{P}\|_{\mr{HS}} \leq \limsup_{t\rightarrow \infty} \|A^t\tilde{P}A^{\ast t}\|_{\mr{HS}} \leq \|\tilde{P}\|_{\mr{HS}}\limsup_{t\rightarrow \infty} \|A^t\| = 0.$ This implies $\tilde{P}=0$. Therefore the solution to \eqref{eq:P.ale} must be unique. 
\end{proof}
\begin{corollary}
    Suppose that the conditions of Lemma \ref{lem:homeomorphism} hold. Then for any $w:\mX\rightarrow \mR_+$ with $w(x)=\ip{\kappa(x, \cdot)}{Q\kappa(x, \cdot)}$ for some $Q\geq 0$, $Q\in \mr{HS}(H_\kappa)$, there exists a unique $v$ specified by $v(x)=\ip{\kappa(x, \cdot)}{P\kappa(x, \cdot)}$ for some $P\geq 0$, $P\in \mr{HS}(H_\kappa)$ such that \eqref{eq:decay} hold. 
\end{corollary}

\section{Data-based Estimation and Its Error}\label{sec:estimation}
\par The Koopman operator can be learned from data by fitting a parametric representation to the evolutionary relation \eqref{eq:push-forward} of the kernel functions. This learning procedure is known as kEDMD \citep{william2015kernel}. 
Suppose we have an i.i.d. sample $\{(x_j, y_j)\}_{j=1}^n$. The kEDMD finds an approximation 
\begin{equation}\label{eq:A.approximation}
\textstyle
    \hat{A}^\ast = \sum_{i,j=1}^n \theta_{ij}\kappa(y_i, \cdot)\times \kappa(x_j, \cdot)    
\end{equation}
where the coefficient matrix $\Theta = [\theta_{ij}]$ solves a least-squares problem on RKHS:
\begin{equation}\label{eq:kEDMD}
\textstyle
    \min \enspace \frac{1}{n}\sum_{j=1}^n \|\hat{A}^\ast \kappa\bra{x_j, \cdot} - \kappa\bra{y_j, \cdot}\|_H^2.   
\end{equation}
It admits an explicit solution of $\Theta = G_{xx}^{-1}$, where the matrix $G_{xx}$ has its entry at the $(i,j)$-th position specified by $\kappa(x_i,x_j)$. 
For later analysis, we provide following conclusion on the error of kEDMD. 
\begin{lemma}[K{\"{o}}hne et al. \cite{kohne2025Linf}, Theorem 5.2]
\label{lem:kEDMD}
    Suppose that the conditions in Lemma \ref{lem:definedness} hold, and let $\hat{A}$ be as in \eqref{eq:kEDMD}, where the sample on $\mX$ has a fill distance of $h_\mr{fill} = \sup_{x\in \mX}\min_{j=1,\dots,n} |x-x_j|$. Then for some constant $c_0$,
    \begin{displaymath}
        \|\hat{A}-A\| \leq c_0h_{\mr{fill}}^{k+1/2}, 
    \end{displaymath} 
    where $k$ is the index in the radial kernel (see Footnote 3). 
\end{lemma}

\subsection{Kernel Ridge Regression for Decay Rate Function}
\par In addition to approximating the Koopman operator $A$ using data, the operator $Q$ that specifies the decay rate function $w$ should also be estimated from data. To this end, suppose that $w$ is a sum of squares of functions on $H_\kappa$: $w(x) = \sum_{i=1}^\infty |\omega_i(x)|^2$, and hence $w(x) = \sum_{i=1}^\infty |\ip{\kappa(x, \cdot)}{\omega_i}|^2$ gives $Q = \sum_{i=1}^\infty \omega_i\times \omega_i$. 
Denote $\epsilon^Q_N = \|\sum_{i=N+1}^\infty \omega_i\times \omega_i\|_{\mr{HS}}$ the truncation error (which equals $0$ if $w$ is an exact sum of squares of $N$ terms), and consider any approximation of $Q_N = \sum_{i=1}^N \omega_i\times \omega_i$ by $\hat{Q}_N = \hat{\omega}_i\times \hat{\omega}_i$. We have  
$$\textstyle 
\|Q-\hat{Q}\|_{\mr{HS}} \leq \sum_{i=1}^N \bra{2\|\omega_i\|\cdot\|\hat{\omega}_i-\omega_i\| + \|\hat{\omega}_i-\omega_i\|^2} + \epsilon^Q_N.$$

\par Approximating $\omega_i$ by $\hat{\omega}_i$ based on data is realized via a well-known routine called \emph{kernel ridge regression} (KRR), where we let $\hat{\omega}_i = \sum_{j=1}^n \alpha_{ij}\kappa(x_j, \cdot)$ and optimize the coefficients by a regularized least-squares problem:
\begin{equation}\label{eq:KRR}
    \textstyle 
    \min_{\{\alpha_{ij}\}_j} \enspace \frac{1}{n}\sum_{k=1}^n \bra{\hat{\omega}_i(x_k) - \omega_i(x_k)}^2 + \lambda_i \|\hat{\omega}_i\|_H^2.  
\end{equation}
Through this routine, we have the approximation of $Q$ as 
\begin{equation}\label{eq:Q.approximation}
    \textstyle \hat{Q} = \sum_{i,j=1}^N \eta_{ij}\kappa(x_i, \cdot)\times \kappa(x_j, \cdot) \text{ with } \eta_{ij} = \sum_{k=1}^N \alpha_{ki}\alpha_{kj}. 
\end{equation}
Since practically the regularization constants $\lambda_i$ can be fine-tuned through validation, we implicitly assume that they can be tuned to the near-optimal order of magnitude. 
The error of KRR was established in the literature and presented as the following lemma. 
\begin{lemma}[Smale \& Zhou \cite{smale2007learning}, Corollary 2]
\label{lem:KRR}
    Suppose that $\mX\subset \mR^d$ is compact and $\omega_i\in H=H_\kappa$. Then given any $\delta\in (0, 1)$, there exists a constant $c>0$ such that $\|\hat{\omega}_i-\omega_i\|_H\leq c \|\omega_i\|_H n^{-1/6} \log(2/\delta)$ holds with probability at least $1-\delta$. 
\end{lemma}
\begin{corollary}
    Under the conditions of Lemma \ref{lem:KRR}, for any $\delta\in (0,1)$, we have
    $$\textstyle 
    \|Q-\hat{Q}\|_{\mr{HS}} \leq c^Q_{n,\delta} \sum_{i=1}^N \|\omega_i\|_{H}^2 + \epsilon^Q_N ,$$
    where $c^Q_{n,\delta} = 2cn^{-1/6}\log(2/\delta) + c^2n^{-1/3}\log^2(2/\delta)$.  
\end{corollary}

\subsection{Estimation of the Lyapunov Function}
\par With kEDMD estimate $\hat{A}$ as in \eqref{eq:A.approximation} and KRR estimate $\hat{Q}$ as in \eqref{eq:Q.approximation}, we end up with a data-based approximation of Lyapunov operator $\hat{P}$, which we require satisfy the operator ALE: $\hat{A}\hat{P}\hat{A}^\ast - \hat{P} = -\hat{Q}$. Since $\mr{range}~\hat{A} = \mr{span}\{\kappa(x_j, \cdot)\}_{j=1}^n$, it becomes obvious that $\hat{P}$ can be represented by a matrix of coefficients $\Pi = [\pi_{ij}]$:
\begin{equation}\label{eq:P.approximation}
\textstyle
    \hat{P} = \sum_{i,j=1}^n \pi_{ij}\kappa(x_i, \cdot)\times \kappa(x_j, \cdot). 
\end{equation}
The coefficients in \eqref{eq:P.approximation} turn out to satisfy the following matrix algebraic Lyapunov equation, which can be deduced by elementary linear algebra manipulations: 
\begin{equation}\label{eq:P.ale.approximation}
\textstyle
\Theta^\top \Gamma^\top \Pi \Gamma \Theta - \Pi = -\mr{H}, 
\end{equation}
where $\Gamma\in \mR^{n\times n}$ has its $(i,j)$-th entry as $\kappa(x_i, y_j)$ and $\mr{H}\in \mR^{n\times n}$ has its $(i,j)$-th entry equal to $\eta_{ij}$ as calculated by \eqref{eq:Q.approximation}.

\subsection{Error Bound on the Lyapunov Function Estimate}
Finally, we combine the operator norm error on $A$ and the HS norm error on $Q$ to provide the guaranteed HS norm error on $P$. In view of $P = \sum_{t=0}^\infty A^tQA^{\ast t}$ and $\hat{P} = \sum_{t=0}^\infty \hat{A}^t Q \hat{A}^{\ast t}$, the error bound is derived from the errors in each term of the infinite series. 
\begin{theorem}\label{th:error.bound}
    Suppose that the conditions of Lemmas \ref{lem:kEDMD} and \ref{lem:KRR} hold, and in addition, $r(\hat{A})<1$. Then there exist constants $c_1, c_2>0$ such that with probability $1-\delta$, we have
    $$\|\hat{P}-P\|_{\mr{HS}}\leq c_1\epsilon_{n, N, \delta} + c_2\epsilon_{n, N, \delta}^2. $$
    where $\epsilon_{n, N, \delta} = \max\BRA{c_0h_{\mr{fill}}^{k+1/2}, c^Q_{n,\delta} \sum_{i=1}^N \|\omega_i\|_{H}^2 + \epsilon^Q_N }$. 
\end{theorem}
\begin{proof}
    For any bounded operator $A$ and HS operator $B$, we have $\|AB\|_{\mr{HS}}\leq \|A\|\cdot \|B\|_{\mr{HS}}$. Hence we verify that 
    \begin{align*}
    \|\hat{P}-P\|_{\mr{HS}} \leq &  \sum_{t=0}^\infty \|A^t\|^2 \|\hat{Q}-Q\|_{\mr{HS}} + \sum_{t=0}^\infty 2\|\hat{A}^t - A^t\| \\ 
    & \cdot (\|\hat{A}^t\|\vee \|A^t\|)(\|\hat{Q}\|_{\mr{HS}} \vee \|Q\|_{\mr{HS}}) 
    \end{align*} 
    Since $r(A)<1$, for any small enough $\epsilon>0$, $\|A^t\| \leq (r(A)+\epsilon)^t$ after some $\overline{t}$ and hence the first series is bounded by a constant multiple of $\epsilon_{n,N,\delta}$. In the second series, $\|\hat{Q}\|_{\mr{HS}} \vee \|Q\|_{\mr{HS}} \leq \mr{const.}+\epsilon_{n,N,\delta}$, and for large enough $t$, $\|\hat{A}^t\|\vee \|A^t\|\leq (r(A)\vee r(\hat{A})+\epsilon)^t$. Moreover, for all $t\geq 1$, 
    $$\textstyle 
    \|\hat{A}^t - A^t\|\leq \|\hat{A}-A\|\sum_{s=0}^{t-1} \|\hat{A}^s\|\cdot \|A^{t-1-s}\|.$$
    When $t$ is large enough, the sum is always bounded by $t(r(A)\vee r(\hat{A})+\epsilon)^t$, which is summable over $t$. As $\|\hat{A}-A\|\leq \epsilon_{n,N,\delta}$, the proof is completed. 
\end{proof}

\section{Numerical Example}\label{sec:numerical}
To demonstrate the above proposed approach to estimate Lyapunov function $v(\cdot)$ as in \eqref{eq:decay}, we consider the Li{\'{e}}nard equation as a numerical example:
$$ \dot{x}_1 = x_2, \enspace \dot{x}_2 = -\frac{x_1}{1+x_1^2} -x_2.$$ 
To convert the continuous-time system into a discrete-time one, we use $\Delta=0.2$ as the sampling time. The decay rate function is specified as $w(x) = |x|^2$. 
While the true Lyapunov function $v$ corresponding to this decay rate does not have an analytical expression, it can be numerically computed by simulation. As an approximation by the locally linearized dynamics using the Jacobian $[0, 1; -1, -1]$, we obtain a quadratic form: $v_{\mr{lin}}(x) = 8x_1^2 + 5x_1x_2 + 5.533x_2^2$. 

\begin{figure}[!t]
    \centering
    \begin{subfigure}[t]{0.44\columnwidth}
        \centering
        \includegraphics[width=\linewidth]{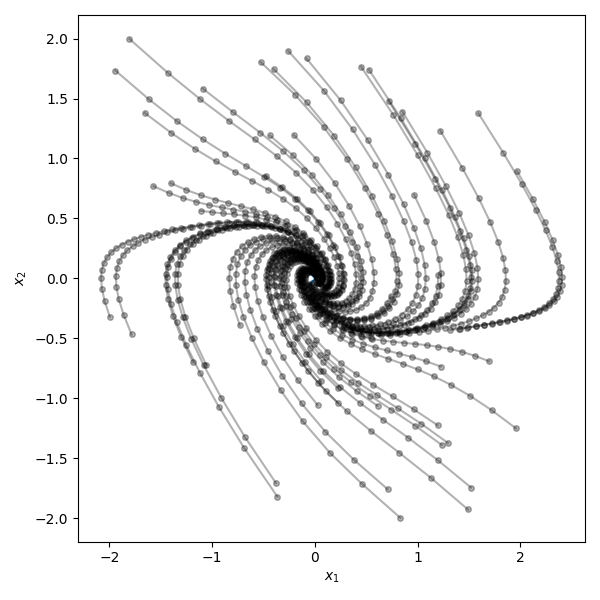}
    \end{subfigure}
    \hfill
    \begin{subfigure}[t]{0.54\columnwidth}
        \centering
        \includegraphics[width=\linewidth]{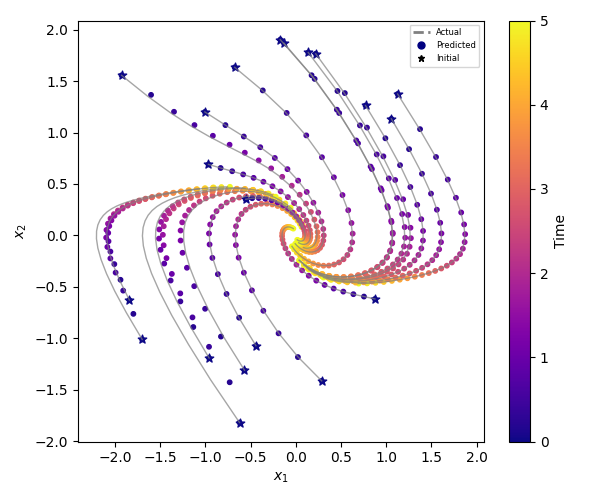}
    \end{subfigure}
    \caption{(Left) Sampled trajectories for Koopman operator learning in kEDMD; (Right) Predicted trajectories using the learned Koopman operator.}
    \label{fig:traj}
    \vspace{-1.5em}
\end{figure}
\par For a data-driven estimation of the Lyapunov function for the nonlinear system without prior knowledge of the model, we sample $50$ trajectories, each over a time horizon of $5$ time units, issued from a random state uniformly distributed in $[-2, 2] \times [-2, 2]$. The sampled trajectories are plotted in the left subfigure in Fig.~\ref{fig:traj}, showing the richness of data. 
Through kEDMD, $\hat{A}$ is learned in the form of \eqref{eq:A.approximation}. The right subfigure of Fig.~\ref{fig:traj} plots the predicted trajectories under the learned operator versus the true orbits. The prediction indeed tends to be accurate. 
In Fig. \ref{fig:spectrum}, the eigenvalues of the estimated Koopman operator, which are equal to those of matrix $\Gamma\Theta$, are plotted and compared to the actual Koopman spectrum, which, in this case, is identical to \eqref{eq:exact.spectrum}. As we observe, the spectrum is confined on the unit disk, which makes it meaningful to estimate the Lyapunov function.
\begin{figure}[!t]
\includegraphics[width=\linewidth]{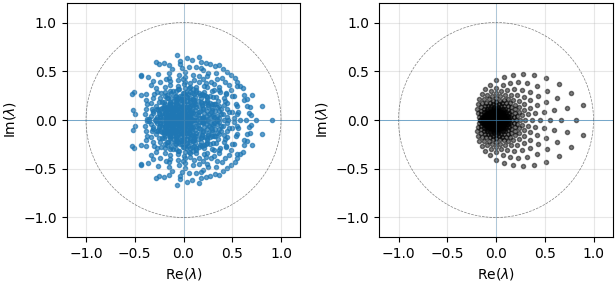}
    \caption{(Left) Eigenvalues of the estimated Koopman operator; (Right) Spectrum of the actual Koopman operator.}
    \label{fig:spectrum}
    \vspace{-1.5em}
\end{figure}

\par With a sum-of-squares decomposition $w = \omega_1^2+\omega_2^2$ with $\omega_1 = e_1$ and $\omega_2=e_2$, and using the dataset, we obtain the KRR estimations $\hat{\omega}_1$, $\hat{\omega}_2$ and hence the approximated operator $\hat{Q}$ that represents the decay rate function $w$ by \eqref{eq:Q.approximation}. 
Subsequently, the coefficients in $\hat{P}$ representing the Lyapunov function $\hat{v}(\cdot)$ is then solved from \eqref{eq:P.ale.approximation}. 
The values of $\hat{v}(x)$ over a square region of $x$ is shown in Fig.~\ref{fig:Lyapunov} as a contour plot in comparison with the true $v(\cdot)$ obtain via simulation. As expected, the Koopman-based estimate overall captures the contour shape of the true Lyapunov function, although the error is more significant at states faraway from the origin due to the fact that Lyapunov function is an \emph{accumulated} decay rate. 
Although the Lyapunov function appear to be ``ellipsoidal'' like a linear system, the nonlinearity can be easily confirmed by the large value deviation in $v(1,1)\approx \hat{v}(1,1)\approx 40$ from $v_{\mr{lin}}(1,1)\approx 19$. 

\par Simulation codes for this numerical experiment are available at this hyperlinked \href{https://github.com/XiuzhenYe/Koopman-based-Estimation-of-Lyapunov-Functions-Theory-on-Reproducing-Kernel-Hilbert-Space}{GitHub repository}.

\begin{figure}[!t]
    \centering
    \begin{subfigure}[t]{0.49\linewidth}
        \centering
        \includegraphics[width=\linewidth]{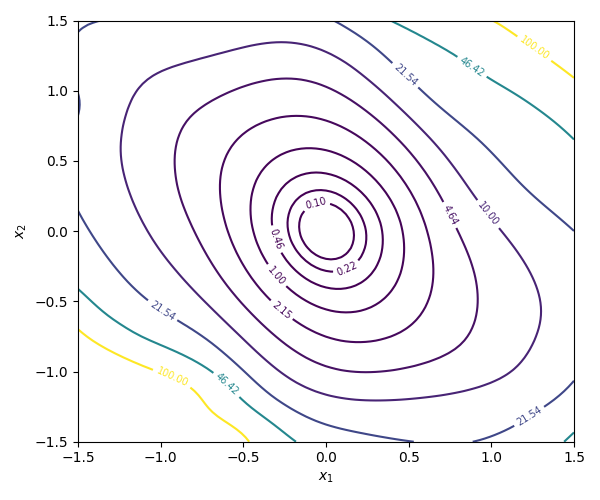}
    \end{subfigure}
    \hfill
    \begin{subfigure}[t]{0.49\linewidth}
        \centering
        \includegraphics[width=\linewidth]{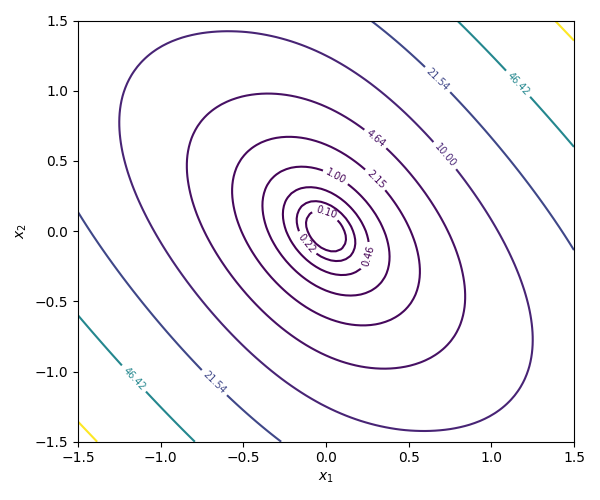}
    \end{subfigure}
    \caption{(Left) Estimated Lyapunov function; (Right) True Lyapunov function.}\label{fig:Lyapunov}
    \vspace{-2em}
\end{figure}

\section{Conclusion}\label{sec:conclusion}
In this paper, we proposed a data-driven approach for estimating the Lyapunov function corresponding to a given decay rate function based on the concept of Koopman operator on a reproducing kernel Hilbert space. 
Specifically, with a linear--radial kernel, when the origin is an asymptotically stable equilibrium point, the Koopman spectral radius is below $1$, which allowing the Lyapunov function to be the unique solution to an operator algebraic Lyapunov equation. 
Under data-based approximation of the Koopman operator, the resulting estimation of the Lyapunov function has an error that vanishes in the large data limit. 

\par  In the circumstances when the origin is not globally asymptotically stable over all $\mX$, the estimation of its domain attraction, or the estimation of all domains of attractions of multiple attractors, naturally become the next step of our research. 
More importantly, the fundamental spectrum--stability relation in a ``Koopman operator on a Hilbert space'' formalism lays a foundation for promising advances in nonlinear control and observation theory, which we consider as a noteworthy direction.

\end{document}